\documentclass[12pt]{amsart}

\usepackage{amsfonts}
\usepackage{amsmath}

\newtheorem{thm}{Theorem}[section]

\newtheorem{prop}[thm]{Proposition}
\newtheorem*{prob*}{Problem}
\newtheorem*{thm*}{Theorem}

\theoremstyle{definition}

\newtheorem*{defn*}{Definition}

\newtheorem*{rem*}{Remark}
\numberwithin{equation}{section}

\newcommand{\C}{\mathbb C}
\newcommand{\R}{\mathbb R}

\DeclareMathOperator{\GAMMA}{Gamma}

\DeclareMathOperator{\per}{per}
\DeclareMathOperator{\gam}{Gamma} \DeclareMathOperator{\sgn}{sgn}

\DeclareMathOperator{\Prob}{Prob}

\begin{document}
\title[Eigenvalues of products of complex Gaussian matrices]
 {\bf{Hole probabilities and overcrowding estimates for  products of complex Gaussian matrices}}

\author{Gernot Akemann}
\address{Department of Physics,  Bielefeld University,  Postfach 100131, D-33501, Bielefeld, Germany} \email{akemann@physik.uni-bielefeld.de}

\author{Eugene Strahov}
\address{Department of Mathematics, The Hebrew University of
Jerusalem, Givat Ram, Jerusalem
91904}\email{strahov@math.huji.ac.il}

\thanks{The first author (G.~A.) is partly supported by the SFB$|$TR12 ``Symmetries
and Universality in Mesoscopic Systems'' of the German research council
DFG. The second author (E.~S.) is supported in part by the US-Israel Binational Science
Foundation (BSF) Grant No.\ 2006333,
 and by the Israel Science Foundation (ISF) Grant No.\ 1441/08.\\ }


\commby{}
\begin{abstract}
We consider eigenvalues of a product of $n$ non-Hermitian, independent random matrices.
Each matrix in this product is of size $N\times N$ with independent standard  complex Gaussian variables.
The eigenvalues of such a product form a determinantal point process on the complex plane (Akemann and Burda \cite{Akemann1}), which can be understood
as a generalization of the finite Ginibre ensemble. As $N\rightarrow\infty$, a generalized infinite Ginibre ensemble arises.
We show that the set of absolute values of the points of this determinantal process has
the same distribution as $\{R_1^{(n)},R_2^{(n)},\ldots\}$, where $R_k^{(n)}$ are independent, and $\left(R_k^{(n)}\right)^2$ is distributed as the product of $n$ independent Gamma variables $\gam(k,1)$.
This enables us to find the asymptotics  for the hole  probabilities, i.e. for the probabilities of the events that there are no points of the process in a  disc of radius $r$ with its center at $0$, as $r\rightarrow\infty$. In addition, we solve the relevant overcrowding problem: we derive an asymptotic formula
for the probability that there are more than $m$ points of the process in a fixed disk of radius $r$ with its center at $0$, as $m\rightarrow\infty$.

\end{abstract}
\maketitle
\section{Introduction}
Products of random matrices are used in different areas of research. For example, the book by Crisanti, Paladin and Vulpiani \cite{Crisanti} describes applications of products of random matrices in  statistical mechanics of disordered systems, localization, wave propagation in random media, and chaotic dynamical systems.
For an application of such products to the study of compositions of random quantum operations we refer the reader to the paper by Roga,  Smaczy$\acute{\mbox{n}}$ski and
$\dot{\mbox{Z}}$yczkowski \cite{Roga}. A paper by Osborn \cite{Osborn}  considers products of random matrices in the context of Quantum Chromodynamics.
For applications to Schr$\ddot{\mbox{o}}$dinger operators see the book by  Bougerol and Lacroix \cite{Bougerol}.

In the 1960's and 70's  different fundamental probabilistic results on  products of random matrices were obtained.
In particular, the asymptotic behavior of products of independent random matrices was investigated in the work by Furstenberg and Kesten \cite{Furstenberg}.
However,  in the 1960's and 70's spectral aspects of  products of random matrices did not attract any serious attention of mathematicians working in the field.
Only  recently a number of works appeared  in which eigenvalue distributions of products of random matrices, in the limit of large  matrices,
were considered (see, for example,  Burda,  Nowak, Jarosz, Livan and Swiech \cite{Burda1,Burda2}, Burda,  Janik, and Waclaw \cite{Burda3}, G\"otze and Tikhomirov \cite{Gotze},
Penson and Zyczkowski \cite{Penson}, O'Rourke and Soshnikov \cite{Soshnikov}, Forrester \cite{Forrester1}), and where
products of random matrices were studied by  usual methods of Random Matrix Theory.  We refer the reader to the books by
Anderson,  Guionnet and Zeitouni \cite{Anderson}, Deift \cite{Deift}, Forrester \cite{Forrester0},  and Pastur and Shcherbina \cite{Pastur} for an introduction to Random Matrix Theory, and for the description of its basic methods and results.

In this article we concentrate on \textit{radial} distributions of eigenvalues of products of complex \textit{non-Hermitian} independent random matrices. For  a properly normalized product
 of complex \textit{non-Hermitian} independent random matrices  O'Rourke and Soshnikov \cite{Soshnikov} showed (under certain assumptions on the entries of the random matrices) that the empirical spectral distribution of the eigenvalues  converges to the limiting
distribution, which is a power of the circular law. Forrester \cite{Forrester1} derived a formula for the Lyapunov exponents for a product of complex Gaussian matrices.
The starting point of the present research is the result obtained in Akemann and Burda \cite{Akemann1}. They considered the product of $n$ complex non-Hermitian, independent random
matrices, each of size $N\times N$  with independent identically distributed Gaussian entries (Ginibre matrices). It was shown that the eigenvalues of such a product
form a complex determinantal point process which can be understood as a generalization of the classical Ginibre ensemble.
It is the aim of the present paper  to study in detail this complex determinantal process (which in this paper is called the \textit{generalized finite-$N$ Ginibre ensemble} with  parameter $n$), and its infinite analogue (which in this paper is called hereafter the \textit{generalized infinite Ginibre ensemble} with  parameter $n$).
 We show that the set of absolute values of the points of the generalized Ginibre ensemble has
the same distribution as $\{R_1^{(n)},R_2^{(n)},\ldots\}$, where $R_k^{(n)}$ are independent, and $\left(R_k^{(n)}\right)^2$ is distributed as the product of $n$ independent Gamma variables $\gam(k,1)$,
see Theorem \ref{TheoremMainResult1} and Theorem \ref{TheoremmainResult2}.
This enables us to find the asymptotics  for the hole  probabilities, i.e. for the probabilities of the events  that there are no points of the process in a  disc of radius $r$ with its center at $0$, as $r\rightarrow\infty$, both for the generalized finite-$N$ and for the generalized infinite  Ginibre ensembles (Theorem \ref{TheoremGapProbabilities}).
In addition, we solve the relevant overcrowding problem: we derive an asymptotic formula
for the probability of the event that there are more than $m$ points of the generalized infinite  Ginibre ensemble in a fixed disk of radius $r$ with its center at $0$, as $m\rightarrow\infty$, see Theorem \ref{TheoremOvercrowdingProblem}.
In proving these Theorems we apply a technique similar to that developed in Kostlan \cite{Kostlan}, Krishnapur \cite{Krishnapur} and Hough, Krishnapur, Peres and Vir$\acute{\mbox{a}}$g \cite{Hough} for the case of the classical Ginibre ensemble, and for the case of random analytic functions.

This paper is organized as follows. In Section \ref{SectionGeneralizedGinibre} we describe the result obtained  in Akemann and Burda \cite{Akemann1}, and define explicitly the relevant determinantal process.  In Section 3 we state  main results of this paper.
 Theorem \ref{TheoremMainResult1} and Theorem \ref{TheoremmainResult2} describe  the distribution of absolute values of the points of the generalized Ginibre ensembles,
 Theorem \ref{PropositionExactFormulae} gives an exact formula for the hole probabilities,
 Theorem \ref{TheoremGapProbabilities} concerns the rate of  the decay of the  hole probabilities, and  Theorem \ref{TheoremOvercrowdingProblem} solves the relevant overcrowding problem.
 The rest of the paper is devoted to proofs of these results.\\
 \textbf{Acknowledgements.}
 Part of this research was conducted during ZIF research program  "Stochastic Dynamics: Mathematical Theory and Applications".
 It is our pleasure to thank the Center for Interdisciplinary Research (ZIF) of Bielfeld University for
hospitality, and the organizers of the ZIF Research Group 2012
"Stochastic Dynamics: Mathematical Theory and Applications"
for the stimulating and encouraging environment they created at the program.

\section{Products of random matrices and generalized Ginibre ensembles}\label{SectionGeneralizedGinibre}
In this article we consider the product $P_n$ of $n$ independent random matrices,
$$
P_n=X_1X_2\ldots X_n.
$$
Each matrix $X_j$, $j=1,\ldots,n$, is of size $N\times N$, and with
i.i.d standard complex Gaussian entries.
Let $z_1$,$\ldots$,$z_N$ be the eigenvalues of $P_n$. We study some statistical properties of the distribution of the eigenvalues $z_1$,$\ldots$,$z_N$
in the complex plane, both for finite and for large $(N\rightarrow\infty)$ matrices. More explicitly, the starting point
of the present work is  the following result obtained recently by  Akemann and Burda \cite{Akemann1}.

Assume that $|z_1|\leq\ldots\leq |z_N|$.
Then the joint density of $(z_i)_{i=1,\ldots,N}$ with respect to Lebesgue measure on $\C^N$ is given by
\begin{equation}\label{AkemannBurdaJointDensity}
\rho_N^{(n)}(z_1,\ldots,z_N)=\left(\frac{1}{\pi^N\prod_{k=1}^{N}\Gamma(k)}\right)^n\prod\limits_{k=1}^Nw_n(z_k)\prod\limits_{1\leq i<j\leq N}|z_i-z_j|^2,
\end{equation}
where
$$
w_n(z)=\pi^{n-1}G_{0,n}^{n,0}\left(|z|^2\biggl|\begin{array}{cccc}
0, & 0, & \ldots, & 0
\end{array}
\right).
$$
Here $G_{0,n}^{n,0}\left(|z|^2\biggl|\begin{array}{cccc}
0, & 0, & \ldots, & 0
\end{array}
\right)$ stands for Meijer's $G$-function with suitable choice of parameters.
For Meijer's $G$-functions we adopt the same notation and definitions as in Luke \cite{Luke},  Gradshtein and Ryzhik \cite{Gradshtein}.
Namely, the Meijer G-function
$G_{p,q}^{m,n}\left(x\biggl|\begin{array}{cccc}
                                                                a_1, & a_2, & \ldots, & a_p \\
                                                                b_1, & b_2, & \ldots, & b_q
                                                              \end{array}
\right)$
is defined as
\begin{equation}
G_{p,q}^{m,n}\left(z\biggl|\begin{array}{cccc}
                                                                a_1, & a_2, & \ldots, & a_p \\
                                                                b_1, & b_2, & \ldots, & b_q
                                                              \end{array}
\right)=\frac{1}{2\pi i}\int\limits_{C}\frac{\prod_{j=1}^m\Gamma(b_j-s)\prod_{j=1}^n\Gamma(1-a_j+s)}{\prod_{j=m+1}^q\Gamma(1-b_j+s)\prod_{j=n+1}^p\Gamma(a_j-s)}z^sds.
\nonumber
\end{equation}
Here an empty product is interpreted as unity, $0\leq m\leq q$, $0\leq n\leq p$, and the parameters
$\{a_k\}$ ($k=1,\ldots,p$) and $\{b_j\}$ ($j=1,\ldots,m$) are such that no pole of $\Gamma(b_j-s)$ coincides with any pole of
$\Gamma(1-a_k+s)$. We assume that $z\in\C\setminus\{0\}$. The contour of integration $C$ goes from $-i\infty$ to $+i\infty$ so that all poles
$\Gamma(b_j-s)$, $j=1,\ldots,m,$ lie to the right of the path, and all poles of $\Gamma(1-a_k+s)$, $k=1,\ldots,n,$ lie to the left of the path.
If $p=0$, then $n=0$, and  we write the corresponding Meijer $G$-function as
$
G_{0,q}^{m,0}\left(x\biggl|\begin{array}{cccc}
 b_1, & b_2, & \ldots, & b_q
\end{array}
\right).
$
In particular, we have
$$
G^{n,0}_{0,n}(t|0,0,\ldots,0)=\frac{1}{2\pi i}\int\limits_{c-i\infty}^{c+i\infty}t^{-s}\Gamma^n(s)ds\;\;(t>0,c>0).
$$
This integral can be evaluated (see Springer and Thompson \cite{Springer0}, Lomnicki \cite{Lomnicki})  by contour integration in the form of an infinite series:
$$
G^{n,0}_{0,n}(t|0,0,\ldots,0)=\sum\limits_{j=0}^{\infty}R(t,n,j),
$$
where $R(t,n,j)$ is the residue of the integrand at the $n$th-order pole
$$
s=-j\;\;(j=0,1,\ldots),
$$
i.e.
$$
R(t,n,j)=\frac{1}{(n-1)!}\frac{d^{n-1}}{ds^{n-1}}\left\{t^{-s}(s+j)^n\Gamma^n(s)\right\}|_{s=-j}.
$$
Specifically, we find
$$
G^{1,0}_{0,1}(t|0,0,\ldots,0)=\sum\limits_{j=0}^{\infty}(-1)^j\frac{t^j}{j!}=e^{-t},
$$
and
$$
G^{2,0}_{0,2}(t|0,0,\ldots,0)=\sum\limits_{j=0}^{\infty}\frac{t^j}{(j!)^2}(-\log t+2\psi(j+1))=2K_0(2\sqrt{t}).
$$
Here $\psi(.)$ is the Euler psi function and $K_0(.)$ is the modified Bessel function of the second kind of zero order, see, for example, Erd$\acute{\mbox{e}}$lyi \cite{Erdelyi}.

Equation (\ref{AkemannBurdaJointDensity}) implies that the eigenvalues of $P_n$ form a determinantal point process\footnote{For a background on determinantal point processes we refer the reader to survey articles by Borodin \cite{Borodin}, and by  Hough, Krishnapur, Peres, and Vir$\acute{\mbox{a}}$g \cite{Hough1}.} on the complex plane with kernel
$$
K^{(n)}_N(z,\xi)=\sum\limits_{k=0}^{N-1}\frac{\left(z\bar{\xi}\,\right)^k}{(k!)^n}
$$
with respect to the background measure $\frac{1}{\pi^n}w_n(z)dm(z)$. Here $dm(z)$ denotes the Lebesgue measure on the complex plane.
The fact that $\frac{1}{\pi^n}w_n(z)dm(z)$ is indeed a probability measure for any positive integer $n$ can be checked using the formula
\begin{equation}\label{MainIntegralFormula}
\begin{split}
\int\limits_{0}^{\infty}t^{\nu-1}G_{p,q}^{m,n}&\left(t\biggl|\begin{array}{cccc}
                                                                a_1, & a_2, & \ldots, & a_p \\
                                                                b_1, & b_2, & \ldots, & b_q
                                                              \end{array}
\right)dt\\
&=\frac{\prod\limits_{j=1}^m\Gamma(b_j+\nu)\prod\limits_{j=1}^n\Gamma(1-a_j-\nu)}{\prod\limits_{j=m+1}^q\Gamma(1-b_j-\nu)\prod\limits_{j=n+1}^p\Gamma(a_j+\nu)},
\end{split}
\end{equation}
see Luke \cite{Luke}, Section 5.6.7.

We will refer to the determinantal point process formed by the eigenvalues of $P_n$ as to the \textit{generalized finite-$N$ Ginibre ensemble} with parameter $n$. The reason is that once there is only one matrix
in the product ($n=1$), we have
$$
w_1(z)=G^{1,0}_{0,1}(|z|^2|0,0,\ldots,0)=e^{-|z|^2},
$$
and equation (\ref{AkemannBurdaJointDensity}) turns into
\begin{equation}
\rho_N^{(n=1)}(z_1,\ldots,z_N)=\frac{1}{\pi^N\prod_{k=1}^{N}\Gamma(k)}\prod\limits_{k=1}^Ne^{-|z_k|^2}\prod\limits_{1\leq i<j\leq N}|z_i-z_j|^2.
\nonumber
\end{equation}
Therefore the determinantal point process formed by eigenvalues of $P_n$  reduces to the classical  Ginibre ensemble at finite $N$, i.e. to the determinant point process
with the kernel
$$
K^{(n=1)}_N(z,\xi)=\sum\limits_{k=0}^{N-1}\frac{\left(z\bar{\xi}\,\right)^k}{k!}
$$
with respect to the background measure $\frac{1}{\pi}e^{-|z|^2}dm(z)$.

If $N=\infty$, then  we call the corresponding determinantal point process the \textit{generalized infinite Ginibre ensemble} with parameter $n$.

\section{Statement of results}
\subsection{The distribution of the moduli of eigenvalues}
Our first result concerns the distribution of the moduli of the eigenvalues of $P_n$.
Recall that gamma variables $\mbox{Gamma}(k,1)$ are those having the following density function
\begin{equation}\label{GammaDensityFunction}
\varrho_k^{(1)}(x)=\left\{
                 \begin{array}{ll}
                   \frac{1}{\Gamma(k)}x^{k-1}e^{-x}, & x\geq 0, \\
                   0, & x<0.
                 \end{array}
               \right.
\end{equation}
\begin{thm}\label{TheoremMainResult1}
Let $P_n=X_1X_2\ldots X_n$  be a product of $n$ independent random matrices. Each matrix $X_j$, $j=1,\ldots,n$, is of size $N\times N$, and with
i.i.d. standard complex Gaussian entries. The set of absolute values of eigenvalues of $P_n$ has the same distribution as the set
$\left\{R_1^{(n)},R_2^{(n)},\ldots,R_N^{(n)}\right\}$, where the random variables $R_1^{(n)},R_2^{(n)},\ldots, R_N^{(n)}$ are independent, and for each $k$, $1\leq k\leq N$, the random variable $\left(R_k^{(n)}\right)^2$ has the same
distribution as the product of $n$ independent and identically distributed gamma variables $\GAMMA(k,1)$.
\end{thm}

The next Theorem gives even more explicit information on the set of random variables $\left\{R_1^{(n)},R_2^{(n)},\ldots,R_N^{(n)}\right\}$.
\begin{thm}\label{TheoremmainResult2}
The random variable $\left(R_k^{(n)}\right)^2$  has the density function given by the formula
\begin{equation}\label{equation g_i(x)}
\varrho_k^{(n)}(x)=\left\{
                 \begin{array}{ll}
                   \frac{1}{(\Gamma(k))^n}G_{0,n}^{n,0}\left(x\biggl|\begin{array}{cccc}
                                                                k-1, & k-1, & \ldots, & k-1
                                                              \end{array}
\right), & x\geq 0, \\
                   0, & x<0.
                 \end{array}
               \right.
\end{equation}
\end{thm}
{\bf Remarks.} \\
(a) Theorem \ref{TheoremMainResult1} generalizes the result obtained by Kostlan \cite{Kostlan} in the case of the classical Ginibre ensemble ($n=1$).\\
(b) The generalized finite-$N$ Ginibre ensemble  with  parameter $n$ is  a determinantal process on $\C$ with the kernel of the form
$K(z,\xi)=\sum\limits_{k}c_k(z\bar\xi\,)^k$ (where $c_k$ are some coefficients)   with respect to a radially symmetric measure. It is a known general fact (see Hough, Krishnapur, Peres and Vir$\acute{\mbox{a}}$g \cite{Hough}, Section 4.7) that the set of absolute values of the points  for such processes  has the same distribution as a set of independent random variables.
Theorem \ref{TheoremMainResult1} and Theorem \ref{TheoremmainResult2} describe this set of independent random variables explicitly.\\
(c) Theorem \ref{TheoremmainResult2} is closely related to the following result on the distribution of a product
of independent gamma variables (see Springer and Thompson \cite{Springer}, Section 3).
\begin{prop}\label{PropositionSpringerThompson}
Let $x_1$, $x_2$, $\ldots$, $x_n$ be $n$ independent gamma variables having density functions
$$
f_k(x_k)=\left\{
           \begin{array}{ll}
             \frac{1}{\Gamma(b_k)}x_k^{b_k-1}e^{-x_k}, & x_k\geq 0, \\
             0, & x_k<0,
           \end{array}
         \right.
$$
where $b_k>0$, $1\leq k\leq n$. Then the probability density function $g(z)$ of the product $z=x_1x_2\ldots x_n$ is  Meijer's $G$-function multiplied
by a normalizing constant, i.e.
$$
g(z)=\frac{1}{\prod\limits_{i=1}^n\Gamma(b_i)}G_{0,n}^{n,0}(z|b_1-1,b_2-1,\ldots,b_n-1).
$$
\end{prop}
\subsection{An exact formula for the hole probabilities}
Denote by $\mathcal{N}_{GG}^{(n)}(r;N)$ the number of points of the generalized finite-$N$  Ginibre ensemble with  parameter $n$ in the disk of radius $r$ with its center at $0$.
Alternatively, $\mathcal{N}_{GG}^{(n)}(r;N)$ can be understood as the number of eigenvalues of the random matrix $P_n$ in the disk of radius $r$ with its center at $0$.
By the hole probability $\Prob\left\{\mathcal{N}_{GG}^{(n)}(r;N)=0\right\}$ we mean the probability of an event that there are no points of the generalized finite-$N$  Ginibre ensemble with  parameter $n$ in the disk of radius $r$ with its center at $0$.
\begin{thm}\label{PropositionExactFormulae}
The hole probability $\Prob\left\{\mathcal{N}_{GG}^{(n)}(r;N)=0\right\}$ for the generalized finite-$N$  Ginibre ensemble with  parameter $n$ can be written as
$$
\Prob\left\{\mathcal{N}_{GG}^{(n)}(r;N)=0\right\}=\prod\limits_{k=1}^N\frac{G_{1,n+1}^{n+1,0}\left(r^2\biggl|\begin{array}{cccc}
                                                                 & 1 &  & \\
                                                                0, & k, & \ldots, & k
                                                              \end{array}
\right)}{\left(\Gamma(k)\right)^{n}}.
$$
\end{thm}
{\bf Remarks.} \\
(a) Note that there is a convenient integral representation for the Meijer $G$-function in the formula above, namely
$$
G_{1,n+1}^{n+1,0}\left(r^2\biggl|\begin{array}{cccc}
                                                                 & 1 &  & \\
                                                                0, & k, & \ldots, & k
                                                              \end{array}
\right)=\frac{1}{2\pi i}\int\limits_{c-i\infty}^{c+i\infty}r^{-2s}\left(\Gamma(k+s)\right)^n\frac{ds}{s},\;\; c>0.
$$
(b) Since
$$
G_{1,2}^{2,0}\left(r^2\biggl|\begin{array}{cc}
                                                                 & 1 \\
                                                                0, & k
                                                              \end{array}
\right)=\Gamma(k,r^2)=\int\limits_{r^2}^{\infty}e^{-t}t^{k-1}dt,
$$
we see that once $n=1$ the formula in the statement of Theorem \ref{PropositionExactFormulae} reduces to
$$
\Prob\left\{\mathcal{N}_{GG}^{(n=1)}(r;N)=0\right\}=\prod\limits_{k=1}^{N}\frac{\Gamma(k,r^2)}{\Gamma(k)}.
$$
The formula just written above for $\Prob\left\{\mathcal{N}_{GG}^{(n=1)}(r;N)=0\right\}$ is well known, see Grobe, Haake, and Sommers \cite{Grobe},  Forrester \cite{Forrester}.

\subsection{The decay of the hole probabilities}
A basic quantity of interest is the decay of the hole probability as $r\rightarrow\infty$. We investigate the decay of the hole probabilities both for the generalized finite-$N$
Ginibre ensemble with  parameter $n$ (the case of products of $n$ random matrices, each of which is of size $N$), and for the infinite generalized  Ginibre ensemble
(the case of product of $n$ infinite random matrices).
\begin{thm}\label{TheoremGapProbabilities}
\textbf{(A)} (The case of products of random matrices of finite size $N$).\\
As $r\rightarrow\infty$,
\begin{equation}
\begin{split}
&\Prob\left\{\mathcal{N}_{GG}^{(n)}(r;N)=0\right\}\\
&=\frac{(2\pi)^{\frac{(n-1)N}{2}}}{n^{\frac{N}{2}}\prod\limits_{k=1}^N(\Gamma(k))^n}
\exp\left\{-nNr^{\frac{2}{n}}
+N\left(N-\frac{1}{n}\right)\log(r)\right\}
\left(1+O\left(r^{-\frac{2}{n}}\right)\right).
\nonumber
\end{split}
\end{equation}
This implies that for the product of  $n$ matrices of  finite size $N$
we have
$$
\underset{r\rightarrow\infty}{\lim}\left(\frac{1}{r^{\frac{2}{n}}}\log\left[\Prob\left\{\mathcal{N}_{GG}^{(n)}(r;N)=0\right\}\right]\right)=-nN.
$$
\textbf{(B)} (The case of products of infinite random matrices).\\
 The following limiting relation holds true
$$
\underset{r\rightarrow\infty}{\lim}\left(\frac{1}{r^{\frac{4}{n}}}\log\left[\Prob\left\{\mathcal{N}_{GG}^{(n)}(r;N=\infty)\right\}=0\right]\right)=-\frac{n}{4}.
$$
\end{thm}
{\bf Remarks.} \\
(a) For the classical (infinite) Ginibre ensemble we have $n=1$, and Theorem \ref{TheoremGapProbabilities}, (B) says that
$$
\frac{1}{r^{4}}\Prob\left\{\log\mathcal{N}_{GG}^{(n=1)}(r;N=\infty)=0\right\}\rightarrow-\frac{1}{4},
$$
as $r\rightarrow\infty$. This asymptotic result for the classical Ginibre ensemble is well known, see Grobe, Haake, and Sommers \cite{Grobe}, Forrester \cite{Forrester},
Hough, Krishnapur, Peres, and Vir$\acute{\mbox{a}}$g \cite{Hough}, Akemann,  Phillips, and Shifrin \cite{Akemann2} for different proofs and related results.\\
(b) The result of Theorem \ref{TheoremGapProbabilities}, (B) can be compared with the decay of hole probabilities for the zeros of the Gaussian  analytic function,
$$
 f(z)=\sum\limits_{n=0}^{\infty}\frac{a_nz^n}{\sqrt{n!}},
$$
where $a_n$ are i.i.d. standard complex Gaussian random variables.
Namely, it was proved by Sodin and Tsirelson \cite{Sodin} that the hole probability for the zeros decays like $\exp\{-cr^4\}$.
Theorem  \ref{TheoremGapProbabilities} says that for the generalized infinite Ginibre ensemble with  parameter $n$  the hole probability decays like $\exp\{-Cr^{\frac{4}{n}}\}$. Thus  we have the same behavior only  for $n=1$.
\subsection{Overcrowding}
 Consider a disk  with a fixed radius $r>0$, and with its center at $0$. Recall that   $\mathcal{N}_{GG}^{(n)}(r;N=\infty)$ denotes the number of points of the generalized infinite  Ginibre ensemble with  parameter $n$ in this disk.
The problem is to estimate the probability of the event that in this disc  there are more than $m$ points of the  ensemble, i.e. to estimate
$\Prob\left\{\mathcal{N}_{GG}^{(n)}(r;N=\infty)\geq m\right\}$. We are especially interested in the decay of this probability, $\Prob\left\{\mathcal{N}_{GG}^{(n)}(r;N=\infty)\geq m\right\}$,
as $m\rightarrow\infty$.

In this article we prove the following Theorem.
\begin{thm}\label{TheoremOvercrowdingProblem}
Let $\mathcal{N}_{GG}^{(n)}(r;N=\infty)$ be the number of points of the infinite generalized Ginibre ensemble in the disk of radius $r$ around $0$.
Then for a fixed $r>0$
$$
\Prob\left\{\mathcal{N}_{GG}^{(n)}(r;N=\infty)\geq m\right\}=\exp\left\{-\frac{1}{2}nm^2\log(m)(1+o(1))\right\},
$$
as $m\rightarrow\infty$.
\end{thm}
{\bf Remarks.} \\
(a) Theorem \ref{TheoremOvercrowdingProblem} is a generalization of the result obtained by Krishnapur (see Krishnapur \cite{Krishnapur}, Section 2.1)
for the classical infinite Ginibre ensemble (the case corresponding to $n=1$).\\
(b) In the context of zeros of Gaussian analytic function the overcrowding problem was formulated by Yuval Peres, and was studied in detail by
 Krishnapur in \cite{Krishnapur}. It was shown that the probability of the event that in the disk with a fixed radius $r$ around $0$  there are more than $m$ zeros of the Gaussian
 analytic function
  decays in the same way as $\Prob\left\{\mathcal{N}_{GG}^{(n=1)}(r;N=\infty)\geq m\right\}$.
 \section{Proofs of Theorem \ref{TheoremMainResult1} and \ref{TheoremmainResult2}}
Let $r_1,\ldots, r_N$ be the moduli of the eigenvalues $z_1,\ldots,z_N$ of $P_n$, i.e. $r_1=|z_1|$, $r_2=|z_2|$, $\ldots$, $r_N=|z_N|$. Thus $r_1\leq\ldots\leq r_N$. We want to find the joint density of $(r_i)_{i=1,\ldots,N}$.
\begin{prop}\label{PropositionTheJointDensityOfModuli}
The joint density of $(r_i)_{i=1,\ldots,N}$ is given by
$$
\frac{(2\pi)^N}{(\pi^N\prod_{k=1}^{N}\Gamma(k))^n}\per[r_i^{2j-1}]_{i,j=1}^N\prod\limits_{j=1}^Nw_n(r_j).
$$
\end{prop}
\begin{proof}Let $z_1,\ldots,z_N$ be the eigenvalues of $P_n$. The joint density of $(z_i)_{i=1,\ldots,N}$
is given by formula (\ref{AkemannBurdaJointDensity}). Set
$$
z_i=r_ie^{i\theta_i},\;\;i=1,\ldots, N.
$$
We have
\begin{equation}
\begin{split}
&\prod\limits_{1\leq i<j\leq N}|z_i-z_j|^2=\left|\sum\limits_{\sigma\in S(N)}(-1)^{\sgn(\sigma)}\prod\limits_{j=1}^Nz_j^{\sigma(j)-1}\right|^2\\
&=\sum\limits_{\sigma,\sigma'\in S(N)}(-1)^{\sgn(\sigma)+\sgn(\sigma')}
\prod\limits_{j=1}^Nr_j^{\sigma(j)-1}e^{i(\sigma(j)-1)\theta_j}
\prod\limits_{k=1}^Nr_k^{\sigma'(k)-1}e^{-i(\sigma'(k)-1)\theta_k}.
\end{split}
\nonumber
\end{equation}
This gives
\begin{equation}
\begin{split}
\int\limits_0^{2\pi}\ldots\int\limits_0^{2\pi}\prod\limits_{1\leq i<j\leq N}|z_i-z_j|^2d\theta_1\ldots d\theta_N=(2\pi)^N\sum\limits_{\sigma\in S(N)}\prod\limits_{j=1}^Nr_j^{2\sigma(j)-2}.
\end{split}
\nonumber
\end{equation}
Therefore the joint density of $(r_i)_{i=1,\ldots,N}$ is
\begin{equation}
\begin{split}
&(2\pi)^N\left(\frac{1}{\pi^N\prod_{k=1}^{N}\Gamma(k)}\right)^n\prod\limits_{k=1}^Nw_n(r_k)\sum\limits_{\sigma\in S(N)}\prod\limits_{j=1}^Nr_j^{2\sigma(j)-2}\prod\limits_{j=1}^Nr_j
\\
&=\frac{(2\pi)^N}{(\pi^N\prod_{k=1}^{N}\Gamma(k))^n}\per[r_i^{2j-1}]_{i,j=1}^N\prod\limits_{j=1}^Nw_n(r_j).
\end{split}
\nonumber
\end{equation}
\end{proof}
Proposition \ref{PropositionTheJointDensityOfModuli} implies that the random variables $y_1=r_1^2$, $y_2=r_2^2$, $\ldots$, $y_N=r_N^2$
have the joint density given by
$$
\per\left[\frac{y_i^{j-1}}{\pi^{n-1}(\Gamma(j))^n}w_n(\sqrt{y_i})\right]_{i,j=1}^N.
$$
Observe that formula (\ref{MainIntegralFormula}) implies
$$
\int\limits_0^{\infty}y^{j-1}w_n(\sqrt{y})dy=\pi^{n-1}\Gamma^n(j).
$$
Therefore the vector of squares of absolute values of eigenvalues of $P_n$ (in uniform order) has the density
$$
\frac{1}{N!}\per\left[\varrho_j^{(n)}(y_i)\right]_{i,j=1}^N,
$$
where the functions $\varrho_j^{(n)}(y)$, $1\leq j\leq N$, are probability density functions defined by
$$
\varrho_j^{(n)}(y)=\left\{
               \begin{array}{ll}
                 \frac{y^{j-1}w_n(\sqrt{y})}{\pi^{n-1}(\Gamma(j))^n}, & y\geq 0, \\
                 0, & y<0.
               \end{array}
             \right.
$$
Since
$$
y^{j-1}w_n(\sqrt{y})=y^{j-1}G_{0,n}^{n,0}(y|0,0,\ldots,0)=G_{0,n}^{n,0}(y|j-1,j-1,\ldots,j-1),
$$
we conclude that $\varrho_j^{(n)}(y)$ can be rewritten as in equation (\ref{equation g_i(x)}).

To complete the proof of Theorem \ref{TheoremmainResult2} we use the following well known fact (see, for example, Kostlan \cite{Kostlan}, Lemma 1.5).
Assume we are given an $N$-tuplet of independent random variables $(A_i)$, $1\leq i\leq N$, with the corresponding densities $(\varrho_i)$, $1\leq i\leq N$.
Define a new $N$-tuplet of random variables, $(B_i)$, $1\leq i\leq N$, as a random permutation of the vector $(A_i)$, $1\leq i\leq N$ (these random permutations are equal
to each other in probability).  Then the joint density of the random vector $(B_i)$, $1\leq i\leq N$, is
$$
\frac{1}{N!}\per\left[\varrho_j(B_i)\right]_{i,j=1}^N.
$$
Considering squares of moduli of unordered eigenvalues of $P_n$ as random variables $(B_i)$, $1\leq i\leq N$, we obtain the statement of Theorem \ref{TheoremmainResult2}.

Theorem \ref{TheoremMainResult1} follows from Theorem \ref{TheoremmainResult2}, and from the  result by Springer and Thompson \cite{Springer} on the distribution of a product
of independent gamma variables, see Proposition \ref{PropositionSpringerThompson}. Namely,
Proposition \ref{PropositionSpringerThompson} and Theorem \ref{TheoremmainResult2} imply that each random variable $\left(R_i^{(n)}\right)^2$, $1\leq i\leq N$, has the the same distribution as the product of $n$ identically distributed and independent gamma variables having density function (\ref{GammaDensityFunction}). Theorem \ref{TheoremMainResult1} is proved.
\qed
\section{Proof of Theorem \ref{PropositionExactFormulae}}
From Theorem \ref{TheoremMainResult1} we conclude that
$$
\Prob\left\{\mathcal{N}_{GG}^{(n)}(r;N)=0\right\}=\prod\limits_{k=1}^N\Prob\left\{\left(R_k^{(n)}\right)^2>r^2\right\},
$$
where the random variables $R_1^{(n)}$, $R_2^{(n)}$, $\ldots$, $R_N^{(n)}$ are those introduced in the statement of Theorem \ref{TheoremMainResult1}.
By Theorem \ref{TheoremmainResult2}
$$
\Prob\left\{\left(R_k^{(n)}\right)^2>r^2\right\}=\frac{1}{(\Gamma(k))^n}\int\limits_{r^2}^{\infty}G_{0,n}^{n,0}(x|k-1,\ldots,k-1)dx.
$$
The last integral can be computed explicitly using the formula
\begin{equation}
\begin{split}
\int\limits_1^{\infty}x^{-\rho}(x-1)^{\sigma-1}G_{p,q}^{m,n}\left(\alpha x\biggl|\begin{array}{ccc}
                                                                 a_1,& \ldots, & a_p  \\
                                                                b_1, &  \ldots, & b_q
                                                              \end{array}
\right)dx=
\Gamma(\sigma)G_{p+1,q+1}^{m+1,n}\left(\alpha \biggl|\begin{array}{cccc}
                                                                 a_1,& \ldots, & a_p, & \rho \\
                                                               \rho-\sigma,& b_1, &  \ldots, & b_q
                                                              \end{array}
\right),
\nonumber
\end{split}
\end{equation}
see Gradshtein and Ryzhik \cite{Gradshtein}, 7.811.3. This gives
$$
\Prob\left\{(R_k^{(n)})^2>r^2\right\}=\frac{r^2G_{1,n+1}^{n+1,0}\left(r^2 \biggl|\begin{array}{cccc}
                                                                 & 0 &  &  \\
                                                              -1,& k-1, &  \ldots, & k-1
                                                              \end{array}
\right)}{(\Gamma(k))^n}.
$$
Since
$$
z^{\sigma}G_{p,q}^{m,n}\left(z \biggl|\begin{array}{ccc}
                                                                 a_1,& \ldots, & a_p \\
                                                               b_1,&   \ldots, & b_q
                                                              \end{array}
\right)=G_{p,q}^{m,n}\left(z \biggl|\begin{array}{ccc}
                                                                 a_1+\sigma,& \ldots, & a_p+\sigma \\
                                                               b_1+\sigma,&   \ldots, & b_q+\sigma
                                                              \end{array}
\right),
$$
(see, for example, Luke \cite{Luke}, Section 5.4) we can rewrite the expression for
$\Prob\left\{\left(R_k^{(n)}\right)^2>r^2\right\}$ as
$$
\Prob\left\{\left(R_k^{(n)}\right)^2>r^2\right\}=\frac{G_{1,n+1}^{n+1,0}\left(r^2\biggl|\begin{array}{cccc}
                                                                 & 1 &  & \\
                                                                0, & k, & \ldots, & k
                                                              \end{array}
\right)}{\left(\Gamma(k)\right)^{n}},
$$
and the result of Theorem \ref{PropositionExactFormulae} follows.
\qed
\section{Proof of Theorem \ref{TheoremGapProbabilities}}
\subsection{Proof of the asymptotic formula for the hole probability for the generalized finite-$N$  Ginibre ensemble}
We use Theorem \ref{PropositionExactFormulae}, which expresses the hole probability
$\Prob\left\{\mathcal{N}_{GG}^{(n)}(r;N)=0\right\}$ in terms of the Meijer $G$-functions.
The following asymptotic formula holds true (see Luke \cite{Luke}, Section 5.7)
\begin{equation}\label{MeijerGAsymptotics}
\begin{split}
&G_{p,q}^{q,0}\left(z\biggl|\begin{array}{cccc}
                                                                a_1, & a_2, & \ldots, & a_p \\
                                                                b_1, & b_2, & \ldots, & b_q
                                                              \end{array}
\right)\sim\frac{(2\pi)^{(\sigma-1)/2}}{\sigma^{1/2}}\exp\left\{-\sigma z^{1/\sigma}\right\}z^{\theta}\sum\limits_{k=0}^{\infty}M_k z^{-k/\sigma},
\end{split}
\nonumber
\end{equation}
where
$$
|z|\rightarrow\infty,\;\;\;|\arg z|\leq(\sigma+\epsilon)\pi-\delta,\;\;\; \delta>0.
$$
In the formula above the parameters $\sigma$ and $\epsilon$ are defined by
$$
\sigma=q-p,
$$
and
$$
\epsilon=\frac{1}{2}\;\;\;\mbox{if}\;\;\;\sigma=1,\;\;\;\epsilon=1\;\;\;\mbox{if}\;\;\;\sigma> 1.
$$
The parameter $\theta$ is defined by the formula
$$
\sigma\theta=\left\{\frac{1}{2}(1-\sigma)+\sum\limits_{k=1}^qb_k-\sum\limits_{k=1}^pa_k\right\}.
$$
Finally, the coefficients $M_k$'s are independent of $z$ and can be found explicitly. In particular, $M_0=1$.

In our case $p=1$, $q=n+1$, $\sigma=n$, and $z=r^2$. It is not hard to find the parameter $\theta$ as well. The result is
$$
\theta=k-\frac{1}{2}-\frac{1}{2n}.
$$
This gives
\begin{equation}\label{MeijerGAsymptotics1}
\begin{split}
&G_{1,n+1}^{n+1,0}\left(r^2\biggl|\begin{array}{cccc}
                                                                 & 1 &  &  \\
                                                                0, & k, & \ldots, & k
                                                              \end{array}
\right)=\frac{(2\pi)^{\frac{n-1}{2}}}{n^{\frac{1}{2}}}\exp\left\{-nr^{\frac{2}{n}}\right\}
r^{2k-1-\frac{1}{n}}\left[1+O\left(r^{-\frac{2}{n}}\right)\right],\\
\end{split}
\nonumber
\end{equation}
as $r\rightarrow\infty$. We insert this asymptotic expression into the formula in the statement of Theorem \ref{PropositionExactFormulae}.
The statement of Theorem \ref{TheoremGapProbabilities} (A) follows immediately.
\qed
\subsection{Proof of the asymptotic formula for the hole probability for the generalized infinite  Ginibre ensemble}
\subsubsection{An upper bound for the hole probability}\label{SectionUpperBound}
To estimate the hole probabilities we use the following standard fact (called the Markov inequality).
\begin{prop}\label{PropositionMarkovInequality}
Suppose $\varphi:\R\rightarrow\R$ is a positive valued function, and let $A$ be a Borel subset of $\R$. Then
$$
\inf\left\{\varphi(y):y\in A\right\}\cdot\Prob\left\{X\in A\right\}\leq \mathbb{E}\varphi(X).
$$
\end{prop}
\begin{proof}
See, for example, Durrett \cite{Durrett}, Section 1.6, Theorem 1.6.4.
\end{proof}
\begin{prop}\label{PropositionUpperBound}
We have
$$
\Prob\left\{\mathcal{N}_{GG}^{(n)}(r;N=\infty)=0\right\}\leq\exp\left\{-\frac{n}{4}r^{\frac{4}{n}}+O\left(\log(r)\right)\right\},
$$
as $r\rightarrow\infty$.
\end{prop}
\begin{proof}
Let $\alpha\geq 0$, $A=(r^2,\infty)$, and $\varphi(x)=x^{\alpha}$. Then by the Markov inequality (Proposition \ref{PropositionMarkovInequality})
we have
$$
\Prob\left\{\left(R_k^{(n)}\right)^2>r^2\right\}\leq\frac{\int\limits_{0}^{\infty}t^{\alpha}G_{0,n}^{n,0}(t|k-1,k-1,\ldots,k-1)dt}{(\Gamma(k))^n(r^2)^{\alpha}}.
$$
(we have used the fact that the random variable $\left(R_k^{(n)}\right)^2$ has the density function given by formula (\ref{equation g_i(x)}), see Theorem \ref{TheoremmainResult2}).
By formula (\ref{MainIntegralFormula})
$$
\int\limits_{0}^{\infty}t^{\alpha}G_{0,n}^{n,0}(t|k-1,k-1,\ldots,k-1)dt=\left(\Gamma(k+\alpha)\right)^n.
$$
Therefore,
$$
\Prob\left\{\left(R_k^{(n)}\right)^2>r^2\right\}\leq\frac{\left(\Gamma(k+\alpha)\right)^n}{(\Gamma(k))^n(r^2)^{\alpha}}.
$$
Next we use the following well known inequality (see, for example, Digital Library of Mathematical Functions \cite{Digital}, $\S$ 5.6)
\begin{equation}\label{GammaInequality}
(2\pi)^{\frac{1}{2}}\exp\left\{-z+(z-\frac{1}{2})\log z\right\}\leq\Gamma(z)\leq(2\pi)^{\frac{1}{2}}\exp\left\{-z+(z-\frac{1}{2})\log(z)+\frac{1}{12 z}\right\},
\end{equation}
as $z\geq 1$,
to obtain
\begin{equation}
\begin{split}
&\Prob\left\{\left(R_k^{(n)}\right)^2>r^2\right\}\\
&\leq\exp\left\{-n\alpha-\alpha\log(r^2)+n\left(k-\frac{1}{2}+\alpha\right)\log\left(1+\frac{\alpha}{k}\right)+\alpha n\log(k)+\frac{n}{12(k+\alpha)}\right\}.
\end{split}
\nonumber
\end{equation}
Set $\alpha=r^{\frac{2}{n}}-k$. Then we can rewrite the inequality above as
\begin{equation}
\begin{split}
\Prob\left\{\left(R_k^{(n)}\right)^2>r^2\right\}&\leq \exp\left\{-n(r^{\frac{2}{n}}-k)+\frac{n}{2}\log\frac{k}{r^{\frac{2}{n}}}-nk\log\frac{k}{r^{\frac{2}{n}}}+\frac{n}{12r^{\frac{2}{n}}}\right\}.
\end{split}
\nonumber
\end{equation}
This gives
\begin{equation}
\begin{split}
&\prod\limits_{k=1}^{\infty}\Prob\left\{\left(R_k^{(n)}\right)^2>r^2\right\}\leq \prod\limits_{k=1}^{r^{\frac{2}{n}}}\Prob\left\{\left(R_k^{(n)}\right)^2>r^2\right\}\\
&\leq\exp\left\{-nr^{\frac{4}{n}}+n\frac{r^{\frac{2}{n}}(r^{\frac{2}{n}}+1)}{2}+\frac{n}{2}\sum\limits_{k=1}^{r^{\frac{2}{n}}}
\log\left(\frac{k}{r^{\frac{2}{n}}}\right)
-n\sum\limits_{k=1}^{r^{\frac{2}{n}}}k\log\left(\frac{k}{r^{\frac{2}{n}}}\right)
+\frac{n}{12}\right\},
\end{split}
\nonumber
\end{equation}
where $r^{\frac{2}{n}}$ is considered as an integer (this assumption should not affect our estimate).
The sums in the exponent can be estimated using the Euler-MacLaurin sum formula. We write it in the form
\begin{equation}\label{EulerMacLaurin}
\sum\limits_{k=1}^{L}f(k)=\int\limits_{1}^{L}f(t)dt+\frac{1}{2}\left(f(L)+f(1)\right)+O\left(f'(L)\right).
\end{equation}
This formula gives
$$
\sum\limits_{k=1}^{r^{\frac{2}{n}}}
\log\left(\frac{k}{r^{\frac{2}{n}}}\right)=-r^{\frac{2}{n}}+1+\frac{1}{n}\log(r)+O\left(r^{-\frac{2}{n}}\right),
$$
and
$$
\sum\limits_{k=1}^{r^{\frac{2}{n}}}k\log\left(\frac{k}{r^{\frac{2}{n}}}\right)=
-\frac{r^{\frac{4}{n}}}{4}+O\left(1\right),
$$
as $r\rightarrow\infty$.
 Using these estimates we find
\begin{equation}
\begin{split}
\Prob\left\{\mathcal{N}_{GG}^{(n)}(r;N=\infty)=0\right\}&=\prod\limits_{k=1}^{\infty}\Prob\left\{\left(R_k^{(n)}\right)^2>r^2\right\}\\
&\leq
\exp\left\{-\frac{n}{4}r^{\frac{4}{n}}+O\left(\log(r)\right)\right\},
\end{split}
\nonumber
\end{equation}
as $r\rightarrow\infty$.
\end{proof}

\subsubsection{A lower bound for the hole probability}\label{SectionLowerBound}
\begin{prop}\label{PropositionLowerBound}
We have
$$
\Prob\left\{\mathcal{N}_{GG}^{(n)}(r;N=\infty)=0\right\}\geq\exp\left\{-\frac{n}{4}r^{\frac{4}{n}}+O(r^{\frac{2}{n}}\log(r))\right\},
$$
as $r\rightarrow\infty$.
\end{prop}
\begin{proof}
It is known that
\begin{equation}\label{LowerBoundn=1}
\prod\limits_{k=1}^{\infty}\Prob\left\{\left(R_k^{(n=1)}\right)^2>r^2\right\}\geq\exp\left\{-\frac{1}{4}r^4+O(r^2\log(r))\right\},
\end{equation}
see, for example, Hough, Krishnapur, Peres and Vir$\acute{\mbox{a}}$g                                                                                       \cite{Hough}, Section 7.2.
Consider the set of the random variables  $\{R_1^{(n)},R_2^{(n)},\ldots \}$. We  know that the random variables $R_k^{(n)}$ are independent, and
$\left(R_k^{(n)}\right)^2$ has the same distribution as the product of $n$ independent and identically distributed gamma variables $\GAMMA(k,1)$.
In particular, the random variable $\left(R_k^{(n=1)}\right)^2$ is itself the gamma variable $\GAMMA(k,1)$. We conclude that
the random variable $\left(R_k^{(n)}\right)^2$ has the same distribution as the random variable $\left[\left(R_k^{(n=1)}\right)^2\right]^n$.
This immediately implies
\begin{equation}\label{LowerBoundInequality}
\Prob\left\{\left(R_k^{(n)}\right)^2>r^2\right\}\geq\left[\Prob\left\{\left(R_k^{(n=1)}\right)^2>r^{\frac{2}{n}}\right\}\right]^n.
\end{equation}
Using inequalities (\ref{LowerBoundn=1}) and (\ref{LowerBoundInequality}) we find
\begin{equation}\label{LowerBoundn}
\begin{split}
\Prob\left\{\mathcal{N}_{GG}^{(n)}(r;N=\infty)=0\right\}&=\prod\limits_{k=1}^{\infty}\Prob\left\{\left(R_k^{(n)}\right)^2>r^2\right\}\\
&\geq\exp\left\{
-\frac{n}{4}r^{\frac{4}{n}}+O(r^{\frac{2}{n}}\log(r))\right\},
\end{split}
\nonumber
\end{equation}
as $r\rightarrow\infty$.
\end{proof}
\subsubsection{Proof of Theorem \ref{TheoremGapProbabilities} (B)}
The statement of Theorem  \ref{TheoremGapProbabilities} (B) follows from Proposition \ref{PropositionUpperBound} and Proposition \ref{PropositionLowerBound}.
\qed
\section{Proof of Theorem \ref{TheoremOvercrowdingProblem}}
Recall that the set of absolute values of points of the generalized Ginibre ensemble has the same distribution as the set
$\{R_1^{(n)},R_2^{(n)},\ldots\}$, where $R_{k}^{(n)}$ are independent, and $\left(R_{k}^{(n)}\right)^2$ has the same distribution as the product of $n$ independent and identically
distributed  gamma variables having density function $\GAMMA(k,1)$, see Theorem \ref{TheoremMainResult1}.
This implies
$$
\left(R_{k}^{(n)}\right)^2\overset{d}{=}\left(\xi^{(1)}_1+\ldots+\xi_k^{(1)}\right)\cdot\ldots\cdot\left(\xi^{(n)}_1+\ldots+\xi_k^{(n)}\right),
$$
where $\xi_i^{(j)}$, $1\leq i\leq k$, $1\leq j\leq n$ are i.i.d. exponential random variables with mean $1$.
Therefore we can write
\begin{equation}
\begin{split}
&\Prob\left\{\left(R_k^{(n)}\right)^2<r^2\right\}
\geq\Prob\left\{\xi_1^{(1)}+\ldots+\xi_k^{(1)}<r^{\frac{2}{n}}\right\}
\cdot\ldots\cdot\Prob\left\{\xi_1^{(n)}+\ldots+\xi_k^{(n)}<r^{\frac{2}{n}}\right\}\\
&\geq\prod\limits_{j=1}^k\Prob\left\{\xi_j^{(1)}<\frac{r^{\frac{2}{n}}}{k}\right\}\cdot\ldots\cdot\prod\limits_{j=1}^k\Prob\left\{\xi_j^{(n)}<\frac{r^{\frac{2}{n}}}{k}\right\}.
\end{split}
\nonumber
\end{equation}
For an exponential random variable $\xi$  with mean $1$ we have
$$
\Prob\left\{\xi<x\right\}\geq\frac{x}{2},\;\;\;0<x<1.
$$
This gives
$$
\Prob\left\{\left(R_k^{(n)}\right)^2<r^2\right\}\geq\left(\frac{r^{\frac{2}{n}}}{2k}\right)^{nk},
$$
and we obtain
\begin{equation}
\begin{split}
&\Prob\left\{{\mathcal{N}}_{GG}^{(n)}(r;N=\infty)\geq m\right\}\geq\prod\limits_{k=1}^m\Prob\left\{\left(R_k^{(n)}\right)^2<r^2\right\}\\
&\geq\prod\limits_{k=1}^m\frac{r^{2k}}{2^{nk}k^{nk}}=\frac{r^{m(m+1)}}{2^{n\frac{m(m+1)}{2}}}\exp\left\{-n\sum\limits_{k=1}^mk\log(k)\right\}.
\end{split}
\nonumber
\end{equation}
The Euler-MacLaurin formula (equation (\ref{EulerMacLaurin})) gives
\begin{equation}\label{KLOGKESTIMATE}
\sum\limits_{k=1}^mk\log(k)=\frac{m(m+1)}{2}\log(m)-\frac{m^2}{4}+O(\log(m)),
\end{equation}
as $m\rightarrow\infty$.
We use this estimate to get a lower bound  for $\Prob\left\{{N}_{GG}^{(n)}(r;N=\infty)\geq m\right\}$, namely
\begin{equation}\label{OverCrowdingLowerBound}
\Prob\left\{{\mathcal{N}}_{GG}^{(n)}(r;N=\infty)\geq m\right\}\geq \exp\left\{-\frac{1}{2}nm^2\log(m)(1+o(1))\right\},
\end{equation}
as $m\rightarrow\infty$.
To obtain an upper bound for $\Prob\left\{{\mathcal{N}}_{GG}^{(n)}(r;N=\infty)\geq m\right\}$ we use the Markov inequality (Proposition \ref{PropositionMarkovInequality}) with
$A=\{0,r^2\}$, $\varphi(x)=x^{-\alpha}$, $\alpha\geq 0$. This gives
\begin{equation}
\begin{split}
\Prob\left\{(R_k^{(n)})^2<r^2\right\}&\leq
(r^2)^{\alpha}\frac{\int\limits_{0}^{\infty}t^{-\alpha}G_{0,n}^{n,0}(t|k-1,k-1,\ldots,k-1)dt}{\left(\Gamma(k)\right)^{n}}\\
&=
(r^2)^{\alpha}\frac{\left(\Gamma(k-\alpha)\right)^n}{(\Gamma(k))^n},
\end{split}
\end{equation}
where we have used formula (\ref{MainIntegralFormula}). By inequality (\ref{GammaInequality})
\begin{equation}
\begin{split}
&\Prob\left\{\left(R_k^{(n)}\right)^2<r^2\right\}\\
&\leq \exp\left\{n\alpha+\alpha\log(r^2)+n(k-\frac{1}{2}-\alpha)\log(1-\frac{\alpha}{k})-\alpha n\log(k)+\frac{n}{12(k-\alpha)}\right\}.
\end{split}
\nonumber
\end{equation}
Choosing $\alpha=k-\frac{1}{2}$, we obtain
\begin{equation}\label{TheUpperBoundInequality}
\begin{split}
&\Prob\left\{\left(R_k^{(n)}\right)^2<r^2\right\}\\
&\leq\exp\left\{\left(k-\frac{1}{2}\right)\left(n+\log(r^2)-n\log(k)\right)+\frac{n}{6}\right\}.
\end{split}
\end{equation}
In addition, by simple probabilistic arguments
\begin{equation}\label{ProbabilisticInequality}
\begin{split}
&\Prob\left\{{\mathcal{N}}_{GG}^{(n)}(r;N=\infty)\geq m\right\}\leq\Prob\left\{\sum\limits_{k=1}^{m^2}\mathbb{I}\left(\left(R_k^{(n)}\right)^2<r^2\right)\geq m\right\}\\
&+\sum\limits_{k=m^2+1}^{\infty}\Prob\left\{\left(R_k^{(n)}\right)^2<r^2\right\}.
\end{split}
\end{equation}
(Here $\mathbb{I}\left(.\right)$ stands for the characteristic function of a set).  The second term on the right hand side of the inequality above can be estimated as follows.
By inequality (\ref{TheUpperBoundInequality})
\begin{equation}
\begin{split}
\Prob\left\{\left(R_k^{(n)}\right)^2<r^2\right\}\leq\exp\left\{-nk\log(k)(1+o(1))\right\},
\nonumber
\end{split}
\end{equation}
as $k\rightarrow\infty$. Therefore,
$$
\sum\limits_{k=m^2+1}^{\infty}\Prob\left\{\left(R_k^{(n)}\right)^2<r^2\right\}\leq \exp\left\{-nm^2\log(m^2)(1+o(1))\right\},
$$
as $m\rightarrow\infty$. Now let us estimate the first term on the right hand side of  inequality (\ref{ProbabilisticInequality}).
We have
$$
\Prob\left\{\sum\limits_{k=1}^{m^2}\mathbb{I}\left(\left(R_k^{(n)}\right)^2<r^2\right)\geq m\right\}\leq\left(\begin{array}{c}
                                                                                                   m^2 \\
                                                                                                   m
                                                                                                 \end{array}
\right)\prod\limits_{k=1}^m\Prob\left\{\left(R_k^{(n)}\right)^2<r^2\right\}.
$$
Using $\left(\begin{array}{c}
                                                                                                   m^2 \\
                                                                                                   m
                                                                                                 \end{array}
\right)<m^{2m}
$, inequality (\ref{TheUpperBoundInequality}), and equation (\ref{KLOGKESTIMATE}) we obtain
\begin{equation}\label{UpperOvercrowdingEstimate}
\begin{split}
&\Prob\left\{\sum\limits_{k=1}^{m^2}\mathbb{I}\left(\left(R_k^{(n)}\right)^2<r^2\right)\geq m\right\}\\
&\leq \exp\left\{2m\log(m)+\frac{nm}{6}+\sum\limits_{k=1}^m(k-\frac{1}{2})\left(n+\log(r^2)-n\log(k)\right)\right\}\\
&=\exp\left\{-\frac{nm^2\log(m)}{2}(1+o(1))\right\},
\end{split}
\end{equation}
as $m\rightarrow\infty$.
The statement of Theorem \ref{TheoremOvercrowdingProblem} follows from inequalities (\ref{OverCrowdingLowerBound}) and (\ref{UpperOvercrowdingEstimate}).
\qed


\begin{thebibliography}{99}
\bibitem{Akemann1}
Akemann, G.; Burda, Z. \textit{Universal microscopic correlation functions for products of independent Ginibre matrices.}
J. Phys. A: Math. Theor. 45 (2012)  465201.
\bibitem{Akemann2}
Akemann, G.; Phillips, M. J.; Shifrin, L. \textit{Gap probabilities in non-Hermitian random matrix theory.} J. Math. Phys. 50 (2009), no. 6, 063504.
\bibitem{Anderson}
Anderson, G. W.; Guionnet, A.; Zeitouni, O. \textit{An introduction to random matrices.} Cambridge Studies in Advanced Mathematics, 118. Cambridge University Press, Cambridge, 2010.
\bibitem{Borodin}
Borodin, A. \textit{Determinantal point processes.} The Oxford handbook of random matrix theory, 231–-249, Oxford Univ. Press, Oxford, 2011.
\bibitem{Bougerol}
Bougerol, P.; Lacroix, J. \textit{Products of random matrices with applications to Schrödinger operators.} Progress in Probability and Statistics, 8. Birkh$\ddot{\mbox{a}}$user Boston, Inc., Boston, MA, 1985.
\bibitem{Burda1}
Burda, Z.; Nowak, M. A.; Jarosz, A.; Livan, G.; Swiech, A. \textit{Eigenvalues and singular values of products of rectangular Gaussian random matrices—the extended version.} Acta Phys. Polon. B 42 (2011), no. 5, 939–-985.
\bibitem{Burda2}
Burda, Z.; Jarosz, A.; Livan, G.; Nowak, M. A.; Swiech, A. \textit{Eigenvalues and singular values of products of rectangular Gaussian random matrices.} Phys. Rev. E (3) 82 (2010), no. 6, 061114.
\bibitem{Burda3}
Burda, Z.; Janik, R. A.; Waclaw, B. \textit{Spectrum of the product of independent random Gaussian matrices.} Phys. Rev. E (3) 81 (2010), no. 4, 041132.
\bibitem{Crisanti}
Crisanti, A.; Paladin, G.; Vulpiani, A..
\textit{Products of Random Matrices in Statistical Physics.} Springer Series in Solid-State Sciences 104. Springer, 2012.
\bibitem{Deift}
Deift, P. A. \textit{Orthogonal polynomials and random matrices: a Riemann-Hilbert approach.} Courant Lecture Notes in Mathematics, 3. New York University, Courant Institute of Mathematical Sciences, New York; American Mathematical Society, Providence, RI, 1999.
\bibitem{Digital}
Digital Library of Mathematical Functions, http://dlmf.nist.gov/
\bibitem{Durrett}
Durrett, R. \textit{Probability: theory and examples.} Fourth edition. Cambridge Series in Statistical and Probabilistic Mathematics. Cambridge University Press, Cambridge, 2010.
\bibitem{Erdelyi}
Erd$\acute{\mbox{e}}$lyi, A.; Magnus, W.; Oberhettinger, F.; Tricomi, F. G. \textit{Higher transcendental functions.} Vol. II. Based on notes left by Harry Bateman. Reprint of the 1953 original. Robert E. Krieger Publishing Co., Inc., Melbourne, Fla., 1981.
\bibitem{Gotze}
G\"otze, F.; Tikhomirov, A. On the Asymptotic Spectrum of Products of Independent Random Matrices.
arXiv:1012.2710v3 [math.PR].
\bibitem{Gradshtein}
Gradshtein, I. S.; Ryzhik, I. M. \textit{Table of integrals, series, and products.} Academic
Press, San Diego 2000.
\bibitem{Grobe}
Grobe, R.; Haake, F.; Sommers, H. J. \textit{Quantum distinction of regular and chaotic dissipative motion.} Phys. Rev. Lett. 61 (1988), no. 17, 1899–-1902.
\bibitem{Forrester0}
Forrester, P. J. \textit{Log-gases and random matrices.} London Mathematical Society Monographs Series, 34. Princeton University Press, Princeton, NJ, 2010.
\bibitem{Forrester}
Forrester, P. J. \textit{Some statistical properties of the eigenvalues of complex random matrices.} Phys. Lett. A 169 (1992), no. 1-2, 21–-24.
\bibitem{Forrester1}Forrester, P. J. \textit{Lyapunov exponents for products of complex Gaussian random matrices.}
arXiv:1206.2001
\bibitem{Furstenberg}
Furstenberg, H.; Kesten, H. \textit{Products of random matrices.} Ann. Math. Statist. 31 1960 457–-469.
\bibitem{Hough}
 Hough, J. B.; Krishnapur, M.; Peres, Y.; Vir$\acute{\mbox{a}}$g, B. \textit{Zeros of Gaussian analytic functions and determinantal point processes.} University Lecture Series, 51. American Mathematical Society, Providence, RI, 2009.
\bibitem{Hough1}
Hough, J. B.; Krishnapur, M.; Peres, Y.; Vir$\acute{\mbox{a}}$g, B. \textit{Determinantal processes and independence}. Probab. Surv. 3 (2006), 206–-229.
\bibitem{Kostlan}
Kostlan, E.  \textit{On the spectra of Gaussian matrices.} Directions in matrix theory (Auburn, AL, 1990). Linear Algebra Appl. 162/164 (1992), 385–-388.
\bibitem{Krishnapur}
Krishnapur, M. \textit{Overcrowding estimates for zeroes of planar and hyperbolic Gaussian analytic functions.} J. Stat. Phys. 124 (2006), no. 6, 1399–1423.
\bibitem{Lomnicki}
Lomnicki, Z. A. \textit{On the distribution of products of random variables.} J. Roy. Statist. Soc. Ser. B 29 (1967), 513-–524.
\bibitem{Luke}
Luke, Y. L.  \textit{The special functions and their approximations.} Academic Press, New York 1969.
\bibitem{Osborn}
Osborn, J. C.
\textit{Universal results from an alternate random-matrix model for QCD with a baryon chemical potential.}
Phys. Rev. Lett. 93 (2004), 222001.
\bibitem{Pastur}
Pastur, L.; Shcherbina, M. \textit{Eigenvalue distribution of large random matrices.} Mathematical Surveys and Monographs, 171. American Mathematical Society, Providence, RI, 2011.
\bibitem{Penson}
Penson, K. A; Zyczkowski, K. \textit{Product of Ginibre matrices: Fuss-Catalan and Raney distributions.} 	Phys. Rev. E (2011) 83, 061118.
\bibitem{Roga}
Roga, W.; Smaczy$\acute{\mbox{n}}$ski, M.;
$\dot{\mbox{Z}}$yczkowski, K.
\textit{Composition of quantum operations and products of random matrices.} (English summary)
Acta Phys. Polon. B 42 (2011), no. 5, 1123–-1140.
\bibitem{Sodin}
Sodin, M.; Tsirelson, B. \textit{Random complex zeroes. III. Decay of the hole probability.} Israel J. Math. 147 (2005), 371–-379.
\bibitem{Soshnikov}
O'Rourke, S.; Soshnikov, A. \textit{Products of independent non-Hermitian random matrices.} Electron. J. Probab. 16 (2011), no. 81, 2219–-2245.
\bibitem{Springer0}
Springer, M. D.; Thompson, W. E.
 \textit{The distribution of products of independent random variables.} SIAM J. Appl. Math. 14 (1966), 511-–526.
\bibitem{Springer}
Springer, M. D.; Thompson, W. E. \textit{The distribution of products of beta, gamma and Gaussian random variables.} SIAM J. Appl. Math. 18 (1970), 721–-737.

\end{thebibliography}
\end{document}